\documentclass[a4paper,twoside]{article}

\usepackage{algorithm}
\usepackage{algorithmic}
\usepackage{graphicx}
\usepackage{epsfig}
\usepackage{subcaption}
\usepackage{calc}
\usepackage{comment}
\usepackage{amssymb}
\usepackage{amstext}
\usepackage{amsmath}
\usepackage{amsthm}
\usepackage{multicol}
\usepackage{pslatex}
\usepackage{apalike}
\usepackage{comment}
\usepackage[bottom]{footmisc}
\usepackage{longtable}
\usepackage[table,xcdraw]{xcolor}
\usepackage{threeparttable}
\usepackage{footmisc}
\usepackage{hyperref}
\usepackage{paralist}
\usepackage{wrapfig}
\usepackage{enumitem}
\usepackage{longtable}
\usepackage[table]{xcolor}
\usepackage{color}
\usepackage{multirow}
\usepackage{caption}

\usepackage{SCITEPRESS}     

\begin{document}

\newtheorem{Definition}{\textbf{Definition}}[]
\newtheorem{mylemma}{\textbf{Lemma}}[]
\newtheorem{myobserve}{\textbf{Observation}}[]
\setlength\parindent{15pt}

\title{Closeness Centrality Algorithms For Homogeneous Multilayer Networks}
\title{Closeness Centrality Detection in Homogeneous Multilayer Networks}

\author{\authorname{Hamza Reza Pavel, Anamitra Roy, Abhishek Santra and Sharma Chakravarthy}
\affiliation {IT Lab, The University of Texas at Arlington, Arlington, Texas 76019}
\email{\{hamzareza.pavel,axr9563,abhishek.santra\}@mavs.uta.edu, sharmac@cse.uta.edu}
}
\abstract{Centrality measures for simple graphs are well-defined and several main-memory algorithms exist for each. Simple graphs have been shown to be not adequate for modeling complex data sets with multiple types of entities and relationships. Although multilayer networks (or MLNs) have been shown to be better suited, there are very few algorithms for centrality measure computation \textit{directly} on MLNs. Typically, they are converted (aggregated or projected) to simple graphs using Boolean AND or OR operators to compute various centrality measures, which is \textit{not only inefficient, but incurs loss of structure and semantics}.\\
In this paper, algorithms have been proposed that compute closeness centrality on a MLN \textit{directly} using a novel decoupling-based approach. Individual results of layers (or simple graphs) of a MLN are used and a composition function is developed to compute the closeness centrality nodes for the MLN. The challenge is to do this efficiently while preserving accuracy of results with respect to the ground truth. However, since these algorithms do not have complete information of the MLN, computing a global measure such as closeness centrality is a challenge. Hence, these algorithms rely on heuristics derived from intuition. The advantage is that this approach lends itself to parallelism and is more efficient compared to the traditional approach. Two heuristics for composition have been presented and their accuracy and efficiency have been empirically validated on a large number of synthetic and real-world-like graphs with diverse characteristics. }


\vspace{-50pt}

\maketitle

\section{Introduction}
\label{sec:introduction}
\noindent Closeness centrality measure, a global graph characteristic, defines the importance of a node in a graph with respect to its distance from all other nodes. Different centrality measures have been defined, both local and global, such as degree centrality  \cite{1}, closeness centrality \cite{cohen2014computing}, eigenvector centrality \cite{3}, multiple stress centrality \cite{4}, betweenness centrality \cite{5}, harmonic centrality \cite{6}, and PageRank centrality \cite{7}. 
Closeness centrality defines the importance of a node based on \textit{how close it is to all other nodes} in the graph. 
Closeness centrality can be used to identify nodes from which communication with \textit{all other nodes in the network} can be accomplished in least number of hops. Most of the centrality measures are defined for simple graphs or monographs. Only page rank centrality has been extended to multi-layer networks \cite{centrality_interconnectedMLN,multiplex_pagerank}. 


A multilayer network consists of layers, where each layer is a simple graph consisting of nodes (entities) and edges (relationships) and optionally connected to other layers through inter-layer edges. If one were to model the three social media networks Facebook, LinkedIn, and Twitter, a MLN is a better model as there are multiple edges (connections) between any two nodes (see Figure~\ref{fig:HoMLN-example}.) This type of MLN is categorized as homogeneous MLNs (or HoMLNs) as the set of entities in each layer has a common subset, but relationships in each layer are different. It is also possible to have MLNs where each layer has different types of entities and relationships within and across layers. Modeling the DBLP data set \cite{dblp}  with authors, papers, and conferences need this type of heterogeneous MLNs (or HeMLNs)~\cite{mln_survey}.

\begin{wrapfigure}{l}{0.45\columnwidth}
    \centering
    \includegraphics[width=0.5\columnwidth]{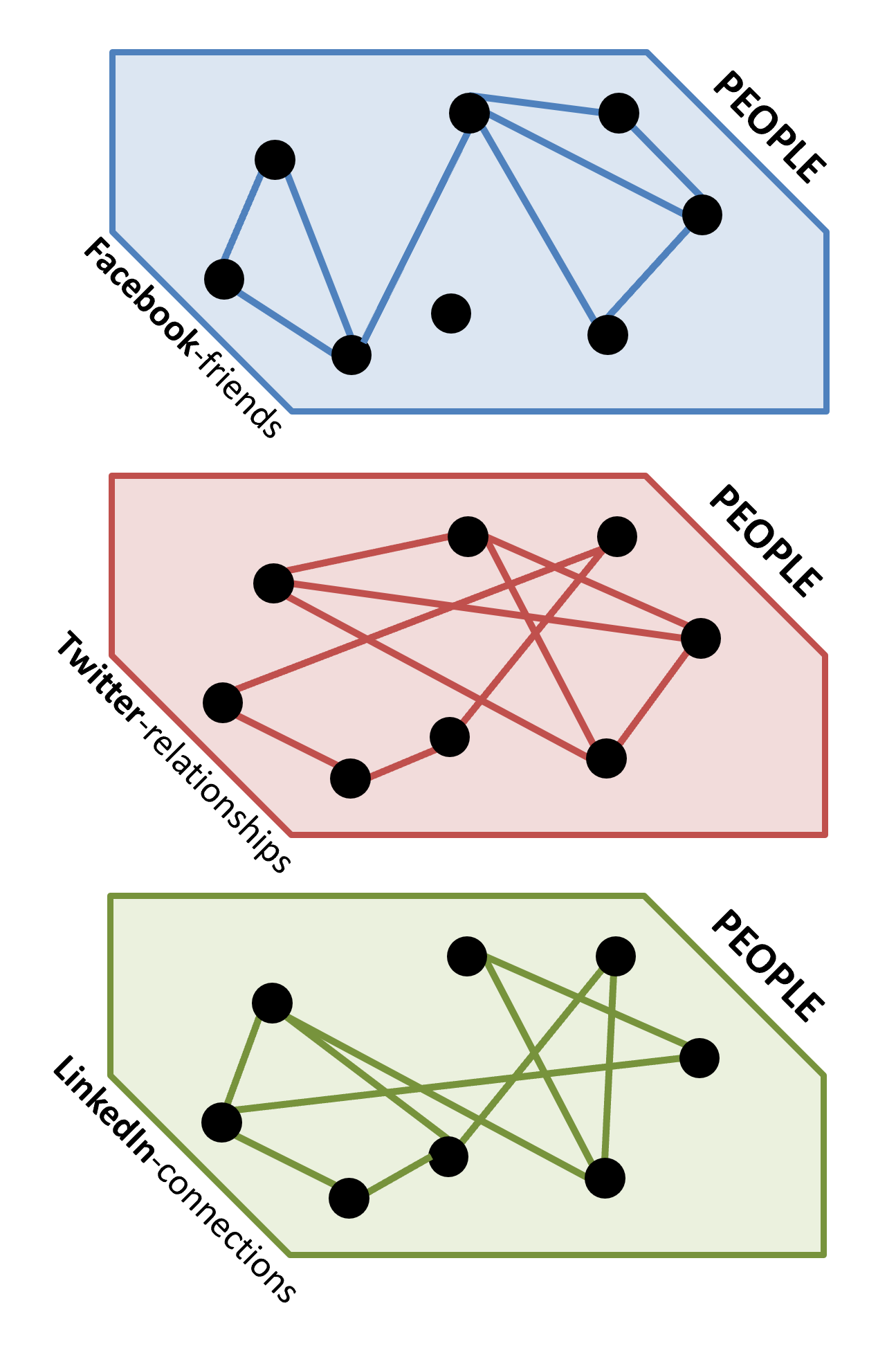}
    \caption{\small HoMLN Example}
    \label{fig:HoMLN-example}
\end{wrapfigure}

This paper presents two heuristic-based algorithms for computing closeness centrality nodes (or CC nodes) on HoMLNs with high accuracy and efficiency. The challenge is in computing a global measure using partitioned graphs (layers in this case) and composing them with minimal additional information to compute a global measure for the combined layers. Boolean AND composition of layers is used for ground truth in this paper (OR is another alternative.) 
These MLN algorithms use the decoupling-based approach proposed in ~\cite{ICCS/SantraBC17,ICDMW/SantraBC17}. Based on this approach, closeness centrality is computed on each layer \textit{once} and minimal additional information is kept from each layer for composing. With this, one can \textit{efficiently} estimate the CC nodes of the MLN. This approach has been shown to be application independent, efficient, lends itself to parallel processing (of each layer), and is flexible for computing centrality measure on any arbitrary subset of layers.

\noindent \textbf{Problem Statement:} Given a homogeneous MLN (HoMLN) with $l$ number of layers --  \textit{G$_1$ (V, E$_1$), G$_2$ (V, E$_2$), ..., G$_l$ (V, E$_l$)}, where $V$ and $E_i$ are the vertex and edge set in the $i^{th}$ layer -- the goal is to identify the closeness centrality nodes of the Boolean AND-aggregated layer consisting of any $r$ layers using the partial results obtained from each of the layers during the analysis step where $r \leq l$. In the decoupling-based approach, the analysis function is defined as the $\Psi$ step and the closeness centrality nodes are estimated in the composition step or the $\Theta$ step using the partial results obtained during the $\Psi$ step.

\subsection{Contributions and Paper Outline}
Contributions of this paper are:
\begin{itemize}
\item Algorithms for computing closeness centrality nodes of MLNs
\item Two heuristics to improve accuracy and efficiency of computed results
\item Use of decoupling-based approach to preserve structure and semantics of MLNs
\item Extensive experimental analysis on large number of synthetic and real-world-like graphs with diverse characteristics.
\item Accuracy and Efficiency comparisons with ground truth and naive approach
\end{itemize}


The rest of the paper is organized as follows: Section \ref{sec:relatedwork} discusses related work. Section \ref{sec:mln_decoupling_approach} introduces the decoupling approach for MLN analysis. Section \ref{sec:mln_algorithm_challenges} provides the challenges of decoupling-based approach for a global metric. Section \ref{sec:closeness_centrality} discusses the challenges in computing the closeness centrality nodes in MLNs. Section \ref{sec:proposedheuristics} describes the proposed heuristics for computing closeness centrality of a MLN. Section \ref{sec:experimental_setup} describes the experimental setup and the data sets. Section \ref{sec:results} discusses result analysis followed by conclusions in Section \ref{sec:conclusions}.


\section{Related Work}
\label{sec:relatedwork}

\noindent Due to the rise in popularity and availability of complex and large real-world-like data sets, there is a critical need for modeling them as graphs and analyzing them in different ways. The centrality measure of MLN provides insight into different aspects of the network. Though there have been a plethora of studies in centrality detection for simple graphs, not many studies have been done on detecting central entities in multilayer networks. Existing studies conducted on detecting central entities in multilayer networks are \textit{use-case specific} and no common framework exists which can be used to address the issue of detecting central entities in a multilayer network.

In \cite{cohen2014computing} proposes an approach to find top k closeness centrality nodes in large graphs. This approximation-based approach has higher efficiency under certain circumstances. Even though the algorithm works on large graphs, unfortunately, it does not work in the case of multi-layer networks. 

In~\cite{centrality_interconnectedMLN}, authors capitalize on the tensor formalism, recently proposed to characterize and investigate complex topologies, to show how closeness centrality and a few other popular centrality measures can be extended to multiplexes. 
The authors in this study \cite{random_walk_centrality} also rely on tensor formalism to investigate and analyze complex multilayer networks. They also extend the random walk closeness centrality to identify nodes in a multilayer network. 

In~\cite{cohen2014computing}, authors propose a sampling-based approach to estimate the closeness centrality nodes with acceptable accuracy. The proposed approach takes linear time and space but it is a \textit{main memory-based algorithm}. The proposed approach works on both undirected and directed graphs. 
The authors of~ \cite{6691611} propose an incremental algorithm that can dynamically update the closeness centrality nodes of a graph in case of edge insertion and deletion. The algorithm has a lower memory footprint compared to traditional closeness centrality algorithms and provides massive speedup when tested on real-world-like graph data sets. 
In~ \cite{du2015new}, the authors propose closeness centrality algorithms where they use effective distance instead of the conventional geographic distance and binary distance obtained by Dijkstra's shortest path algorithm. This approach works on directed, undirected, weighted, and unweighted graphs. 
In~ \cite{9073040}, authors propose an approach to compute the exact closeness centrality values of all nodes in dynamically evolving directed and weighted networks. The proposed approach is parallelizable and achieves a speedup of up to 33 times. 

Most of the methods to calculate closeness centrality are \textbf{main memory-based and not suitable for large graphs}. In \cite{ICDMW/SantraBC17}, the authors propose a decoupling-based approach where each layer can be analyzed independently and in parallel and calculate graph properties for a HoMLN using the information  obtained for each layer.  The proposed algorithms are based on the network decoupling approach which has been shown to be efficient, flexible, scalable as well as accurate. 



\section{Decoupling Approach For MLNs}
\label{sec:mln_decoupling_approach}

\noindent Most of the algorithms available to analyze simple graphs for centrality, community, and substructure detection cannot be used for MLN analysis \textit{directly}. There have been some studies that extend existing algorithms for centrality detection (e.g., page rank) to MLNs \cite{centrality_interconnectedMLN}, but they try to \textbf{work on the MLN as a whole}. The network decoupling approach~\cite{ICCS/SantraBC17,ICDMW/SantraBC17} used in  this paper not only uses extant algorithms for simple graphs, but also uses a partitioning approach for efficiency, flexibility, and new algorithm development.

Briefly, existing approaches for multilayer network analysis convert or transform a MLN into a single graph. This is done either by aggregating or projecting the network layers into a single graph. For homogeneous MLNs, edge aggregation is used to aggregate the network into a single graph. 
Although aggregation of a MLN into a single graph allows one to use extant algorithms (and there are many of them), due to aggregation, \textbf{structure and semantics of the MLN is not preserved resulting in information loss}. 


The network decoupling approach is shown in Figure~\ref{fig:decoupling}.
It consists of identifying two functions: one for analysis ($\Psi$) and one for composition ($\Theta$). Using the analysis function, each layer is analyzed independently (and in parallel). The results (which are termed partial from the MLN perspective) from each of the two layers  are then combined using a composition function/algorithm to produce the results for the two layers of the HoMLN. This binary composition can be applied to MLNs with more than two layers. Independent analysis allows one to use existing algorithms on smaller graphs. The decoupling approach, moreover, adds efficiency, flexibility, and scalability.

\begin{figure}[h]
  \centering
  \includegraphics[width=\linewidth]{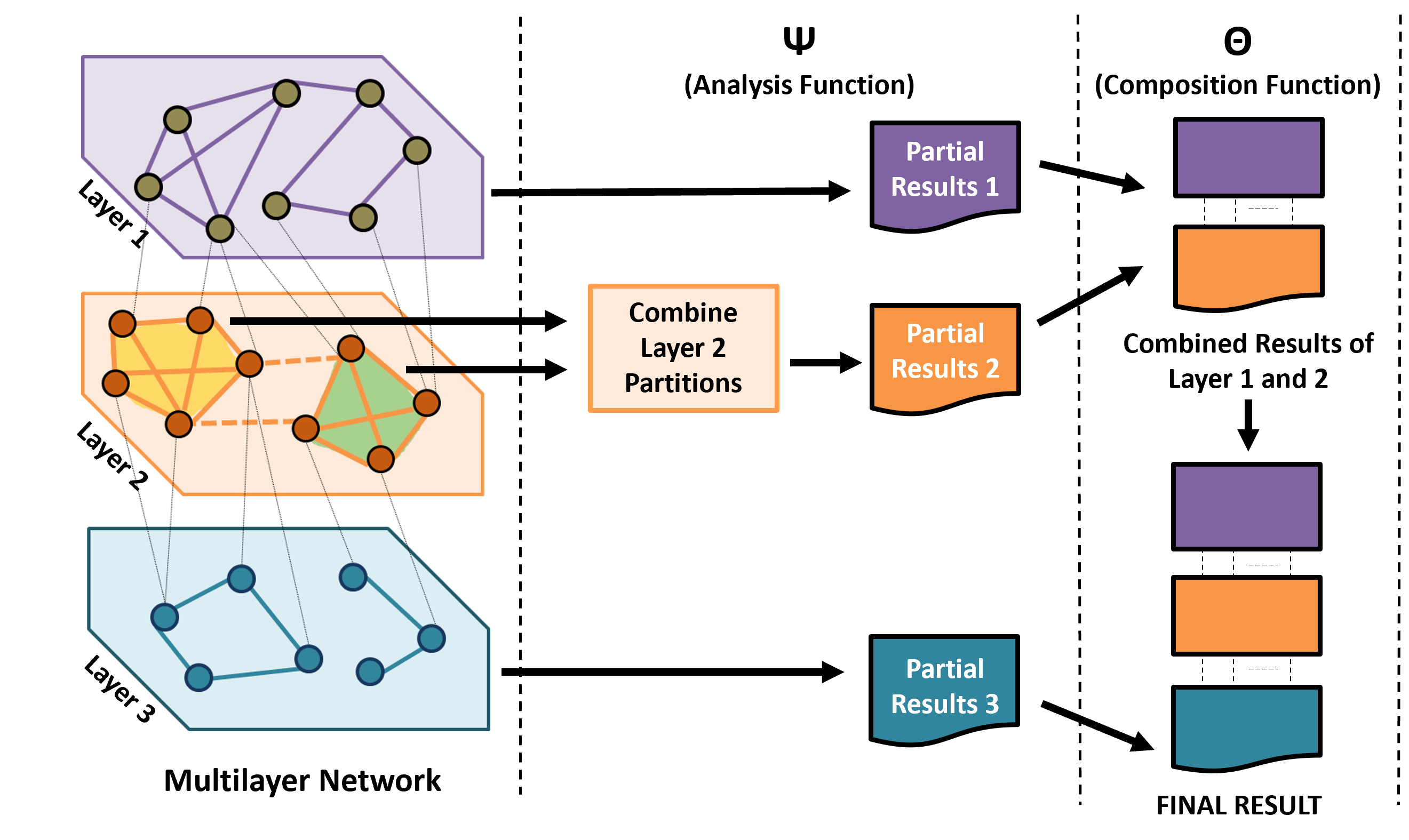}
  \caption{Overview of the network decoupling approach.}
  \label{fig:decoupling}
\end{figure}

The network decoupling method preserves structure and semantics which is critical for drill-down and visualization of results. As each layer is analyzed independently, the analysis can be done in parallel reducing overall response time. Due to the MLN model, each layer (or graph) is likely to be smaller, requires less memory than the entire MLN, and provides clarity. The results of the analysis functions are saved and used for the composition. Each layer is analyzed only once. Typically, the composition function is less complex and is quite efficient, as shall be shown.
Any of the existing simple graph centrality algorithms can be used for the analysis of individual layers. \textit{Also, this approach is application-independent. }

When compared to similar single network approaches, achieving high accuracy with a decoupling approach is the challenge, especially for global measures. While analyzing one layer, identifying minimal additional information needed for improving accuracy due to composition is the main challenge. For many algorithms that have been investigated, there is a direct relation between the amount of additional information used and accuracy gained, however this effects the efficiency.

\subsection{Benefits and Challenges of Decoupling-Based Approach}
\label{sec:mln_algorithm_challenges}

\noindent For analyzing MLNs, currently, HoMLNs are converted into a single graph using aggregation approaches. Given two vertices $u$ and $v$, the edges between them are aggregated into a single graph. The presence of an edge between vertices $u$ and $v$ depends on the aggregation function used. In Boolean AND composed layers, if an edge is present between the same vertex pair $u$ and $v$ in \textit{both} layers, then it will be present in the AND composed layer. Similarly, in Boolean OR composed layers, if an edge is present between the same vertex pair $u$ and $v$ in \textit{at least one} of the layers in HoMLN, then the edge will be present in the OR composed layer. 

Both HoMLNs and HeMLNs are a set of layers of single graphs. Hence, the MLN model provides a natural partitioning of a large graph into layers of a MLN. 
The layer-wise analysis as the basis of the decoupling approach has several benefits. First, the entire network need not be loaded into memory, only a smaller layer. Second, the analysis of the individual layers can be parallelized decreasing the total response time of the algorithm. Finally, the computation used in the composition function ($\Theta$) is based on intuition which is embedded into the heuristic and requires significantly less computation than $\Psi$.

When analyzing a MLN, the accuracy depends on the information being kept (in addition to the output) during the analysis of individual layers. In terms of centrality measures, the bare minimum information that can be kept from each layer is the high centrality (greater than average centrality value) nodes of that layer along with their centrality values. 
Retaining the minimal information, \textit{local centrality measures} such as degree centrality can be calculated relatively easily with high accuracy \cite{kdir22} \cite{pavel2022degree}. 


However,  calculation of a global measure, such as closeness centrality nodes, requires information of the entire MLN. \textit{This compounds the difficulty of computation of closeness centrality of an MLN in the decoupling approach partial information used for estimation of the result will greatly impact the accuracy. Identification of useful minimal information and the intuition behind that are the primary challenges.}
\section{Closeness Centrality: Challenges}
\label{sec:closeness_centrality}

\noindent The closeness centrality value of a node $v$ describes how far are the other nodes in the network from $v$ or how fast or efficiently a node can spread information through the network. For example, when an internet service provider considers choosing a new geo-location for their servers, they might consider a city that is geographically closer to most cities in the region. An airline is interested in identifying a city for their hub that connects to other important cities with a minimum number of hops or layovers. 
For both, computing the closeness centrality of the network is the answer.

\begin{figure}[h]
  \centering
  \includegraphics[width=\linewidth]{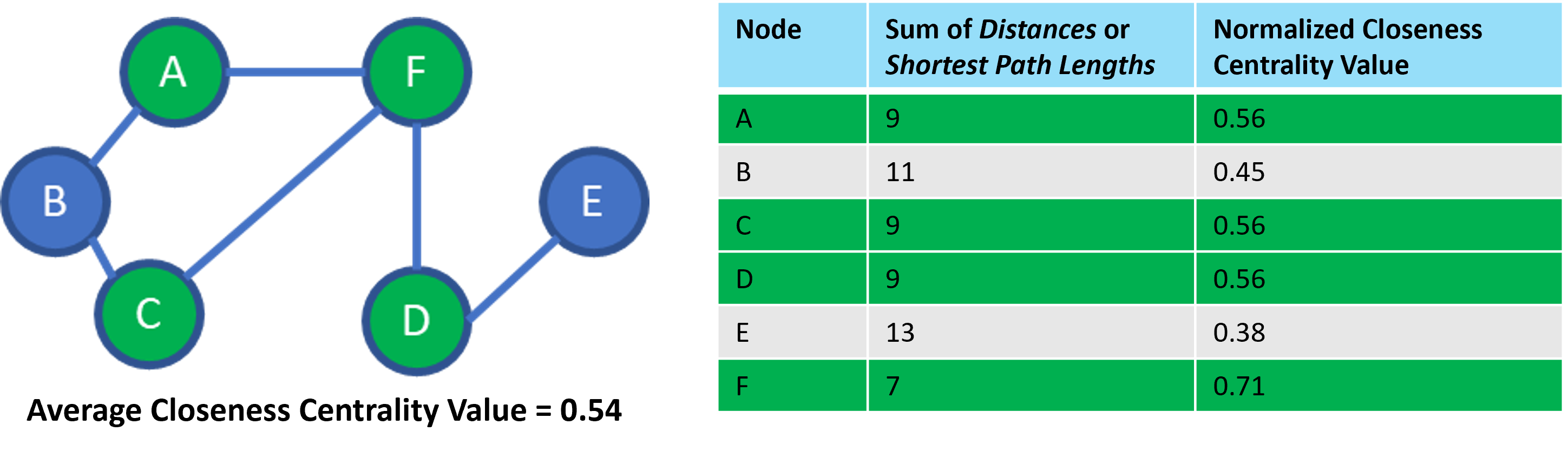}
  \caption{Closeness centrality for a small toy graph.}
  \label{fig1}
\end{figure}

\begin{figure*}[h]
    \centering
    \includegraphics[width=1.6\columnwidth]{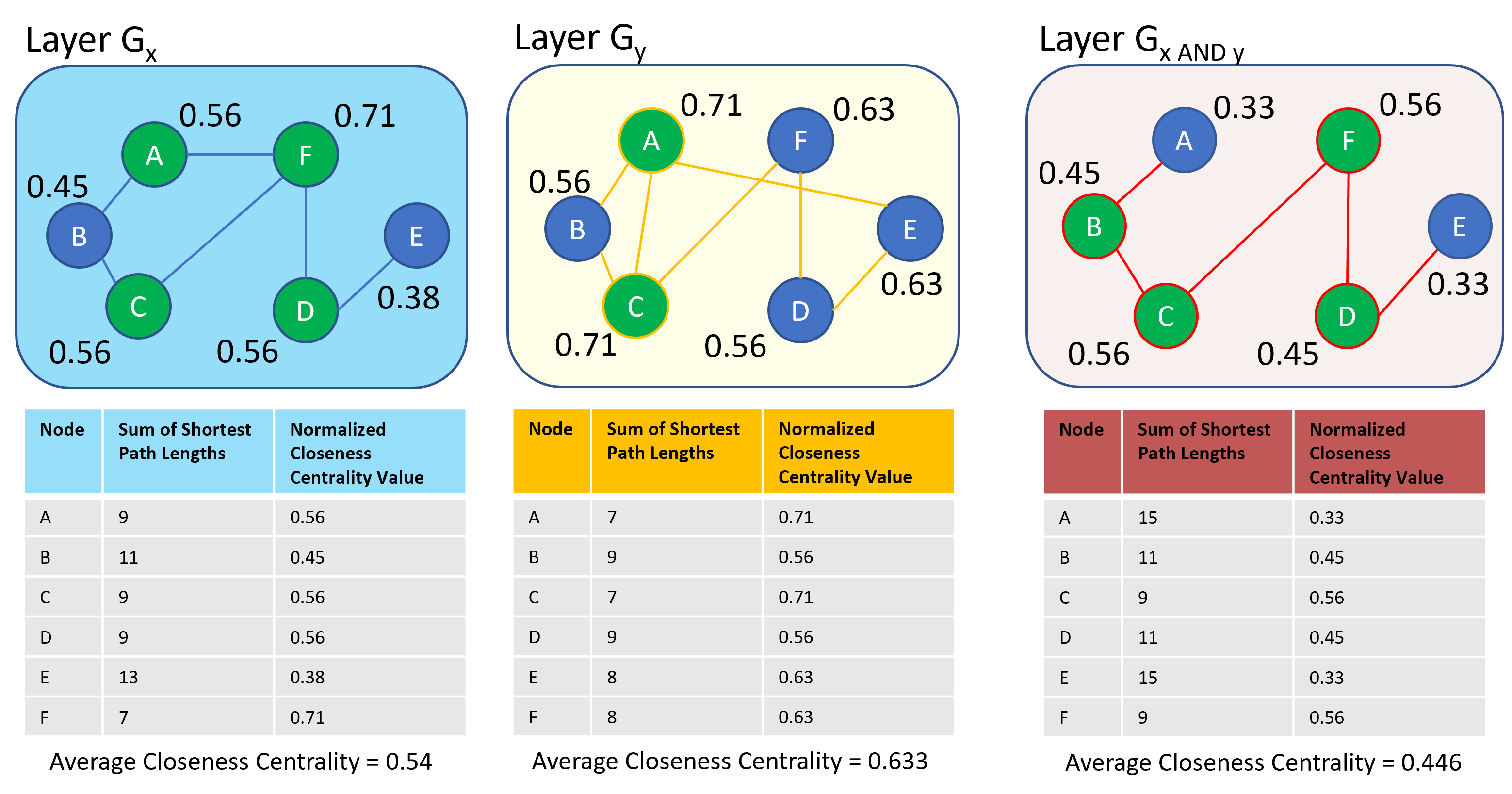}
    \caption{\small Layer $G_x$, $G_y$, and AND-aggregated layer created by $G_x$ and $G_y$, $G_{xANDy}$ (which is computed as $G_x~AND~G_y$) with the respelctive closeness centrality of each node. The nodes highlighted in green have above average closeness centrality values in their respective layers.}
    \label{fig:homln-lemma}
\end{figure*}

The closeness centrality score/value of a vertex $u$ in a network is defined as, 
\begin{equation}
\label{eq:1}
CC(u) =  \frac{n-1}{\sum_{v}d(u, v))}
\end{equation} 
where $n$ is the total number of nodes and $d(u,v)$ is the shortest distance from node $u$ to some other node $v$ in the network. \textit{The higher the closeness centrality score of a node, the more closely (distance-wise) that node is connected to every other node in the network.}
Nodes with a closeness centrality score higher than the average are considered \textbf{closeness centrality nodes (or CC nodes)}. This definition of closeness centrality is only defined for graphs with a \textit{single connected component}. The closeness centrality is defined for both directed and undirected graphs. In this paper, for an algorithm based on the decoupling approach, the focus is on the problem of finding high (same or above average) closeness centrality nodes of Boolean AND aggregated layers of a MLN for undirected graphs. Even though closeness centrality is not well defined for graphs with multiple connected components, the heuristics work for networks where each layer could consist of multiple connected components or the AND aggregated layer has multiple connected components. The proposed heuristics consider the normalized closeness centrality values over the connected component in the layers \cite{wasserman1994social}.

For closeness centrality discussed in this paper, the \textbf{ground truth} is calculated as follows: i) Two layers of the MLN are aggregated into a single graph using the Boolean AND operator and ii) Closeness centrality nodes of the aggregated graph are calculated using an existing algorithm. \textit{The same algorithm is also used on the layers for calculating CC nodes of each layer.}

For finding the ground truth CC nodes and identifying the CC nodes in the layers, the NetworkX package \cite{hagberg2008exploring} implementation of closeness centrality \cite{freeman1978centrality} \cite{wasserman1994social} is used. The implementation of the closeness centrality algorithm in this package uses breadth-first search (BFS) to find the distance from each vertex to every other vertex. For disconnected graphs, if a node is unreachable, a distance of $0$ is assumed and finally, the obtained scores are normalized using Wasserman and Faust approximation which prioritizes the closeness centrality score of vertices in larger connected components \cite{wasserman1994social}. For a graph with V vertices and E edges, the time complexity of the algorithm is $O(V(V+E))$. 

For two-layer MLN, the \underline{\textbf{naive composition}} (as a base composition algorithm for comparison) amounts to taking the CC nodes of each layer (using the same algorithm used for the ground truth), followed by their intersection (as AND aggregation is being used). The resultant set of nodes will be the estimated CC nodes of the MLN layers using the composition approach. The naive approach is the simplest form of composition (using the decoupling approach) and does not use any additional information other than the CC nodes from each layer. The hypothesis is that the naive approach is going to perform poorly when the topology of the two layers is very different. Observation~\ref{obs:1} illustrates that the naive composition approach is \textbf{not guaranteed} to give ground truth accuracy, due to the \textit{generation of false positives and false negatives.} One situation where naive accuracy coincides with the ground truth accuracy is when the two layers are identical. The naive accuracy will fluctuate with respect to ground truth without reaching it in general. \\

\begin{myobserve}
\label{obs:1}
\textit{A node that has above average closeness centrality value in the AND-aggregated layer created by $G_x$ and $G_y$ is not guaranteed to have above average closeness centrality value in one or both of the layers $G_x$ and $G_y$}.
\end{myobserve}

\noindent \textbf{Example}: For single connected component graphs, assume that an arbitrary node, say $u$, \textit{does not} have above average centrality value in layer $G_x$ and layer $G_y$. This means that node $u$ has longer paths to reach every other node in the individual networks. That is, there exist other nodes, say $v$, which cover the entire \textit{individual} networks through shorter paths, thus having high closeness centrality value. There can be scenarios in which these other nodes ($v$) \textit{do not have enough short paths} that exist in \textit{both} layer $G_x$ and layer $G_y$, and as a result bring down their closeness centrality values (and also the average) in the aggregated layer, $G_{xANDy}$. Here, if the node $u$, has common paths from the two layers that are shorter than the common paths for other nodes, then its closeness centrality value will go above average in the AND aggregated layer.  

This scenario is exemplified by node \textbf{B} in Figure \ref{fig:homln-lemma}, which in spite of not having above average closeness centrality value in layers $G_x$ and $G_y$, has above average closeness centrality value in the AND-aggregated layer, $G_{xANDy}$. It can be observed that the closeness centrality value for other nodes decreased significantly as they did not have enough common shorter paths, however, node \textbf{B}, maintained enough common short paths, thus pushing its closeness centrality value above the average in the AND aggregated layer.
 
Hence, it illustrates that a node that has above average closeness centrality value in the AND-aggregated layer created by $G_x$ and $G_y$ is not guaranteed to have above average closeness centrality value in one or both of the layers $G_x$ and $G_y$.

\begin{mylemma}
\label{lem:4}
\textit{It is sufficient to maintain, for every pair of nodes, \textbf{all paths} from the layers, $G_x$ and $G_y$, in order to find out the shortest path between the same nodes in the AND-aggregated layer created by $G_x$ and $G_y$.}
\end{mylemma}

\begin{proof}
Suppose, for every pair of nodes, say $u$ and $v$, the set of \textit{all} paths from the layers $G_x$ and $G_y$, say $P_{x}(u,v)$ and $P_{y}(u,v)$, respectively, is being maintained. Then, by set intersection between $P_{x}(u,v)$ and $P_{y}(u,v)$, one can find out the shortest among the paths that exist between $u$ and $v$ in \textit{both} the layers, $G_x$ and $G_y$, which is basically the shortest path between them in the ANDed graph ($G_{xANDy}$) as well. The common path has to be part of the AND-aggregated graph by definition. If there are no common paths, the AND-aggregation is disconnected and the result still holds.
Figure \ref{fig:homln-lemma} shows an example. For the nodes \textbf{A} and \textbf{F}, if it is maintained $P_{x}(A,F)$ = $\{<$(A,F)$>$, $<$(A,B), (B,C), (C,F)$>\}$ from layer $G_x$ and $P_{y}(A,F)$ = $\{<$(A,C), (C,F)$>$, $<$(A,C), (B,C), (C,F)$>$, $<$(A,E), (E,D), (D,F)$>\}$ from layer $G_y$, then one can obtain the path $<$(A,C), (B,C), (C,F)$>$ as the shortest common path, which is the correct shortest path between these nodes in the ANDed layer, $G_{xANDy}$. The shortest path from $G_x$ between A and F cannot appear in the ANDed graph. This proves the Lemma~\ref{lem:4} that it is sufficient to maintain all paths between every pair of nodes to find out the shortest paths in the ANDed layer.
\end{proof}

Based on Lemma~\ref{lem:4}, 
in the composition step, for every pair of nodes, the path sets from two layers need to be intersected to find out the shortest common path. In doing so the closeness centrality value for each node in the ANDed graph can be calculated. Clearly, if $G_x (V, E_x)$ and $G_y (V, E_y)$ are two layers then the number of paths between any two nodes will be $(V-2)!$, in the worst case. Thus, the composition phase will have a complexity of $O((V-2)!(V-2)!)$, which is \textit{exponentially higher} than the ground truth complexity $O(V(V+min(E_x, E_y)))$ and defeats the entire purpose of the decoupling approach. Thus, \textbf{the challenge is to identify the minimum amount of information to gain the highest possible accuracy over the naive approach by reducing the number of false positives and false negatives, without compromising on the efficiency.}

\section{Closeness Centrality Heuristics}
\label{sec:proposedheuristics}
Two heuristic-based algorithms have been proposed for computing CC nodes for Boolean AND aggregated layers using the decoupled approach. The accuracy and performance of the algorithms have been tested against the ground truth. Extensive experiments have been performed on data sets with varying graph characteristics to show that the solutions work for any graph and have much better accuracy than the naive approach. Also, the efficiency of the algorithms is significantly better than ground truth computation. The solution can be extended to MLNs with any number of layers, where the analysis phase is applied once and the composition phase is applied as a function of pairs of partial layer results, iteratively. 

Jaccard coefficient is used as the measure to compare the accuracy of the solutions with the ground truth. Precision, recall, and F1-score have also been used as evaluation metrics to compare the accuracy of the solutions. 
For performance, the execution time of the solution is compared against that of the ground truth. The ground truth execution time is computed as: time required to aggregate the layers using AND composition function + time required to identify the CC nodes on the combined graph. The time required for the proposed algorithms using the decoupling approach is: \textbf{max}(layer 1, layer 2 analysis times) + composition time. \textbf{The efficiency results do not change much even if the max is not used.}




\subsection{Closeness Centrality Heuristic CC1}

\begin{figure*}[t!]
    \centering
    \begin{subfigure}[t]{0.4\textwidth}
        \centering
        \includegraphics[width=\textwidth]{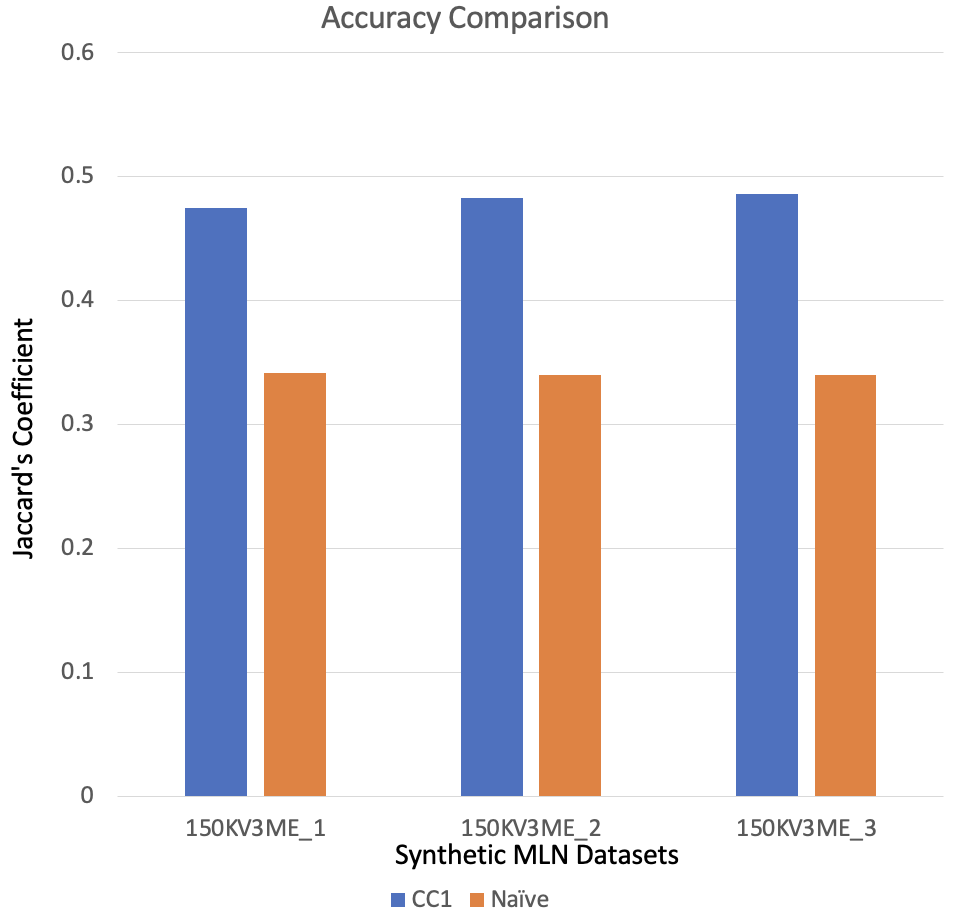}
        \caption{Accuracy.: \texttt{CC1} vs. Naive}
        \label{fig:cc1_accuracy_sec5}
    \end{subfigure}%
    ~ 
    \begin{subfigure}[t]{0.4\textwidth}
        \centering
        \includegraphics[width=\textwidth]{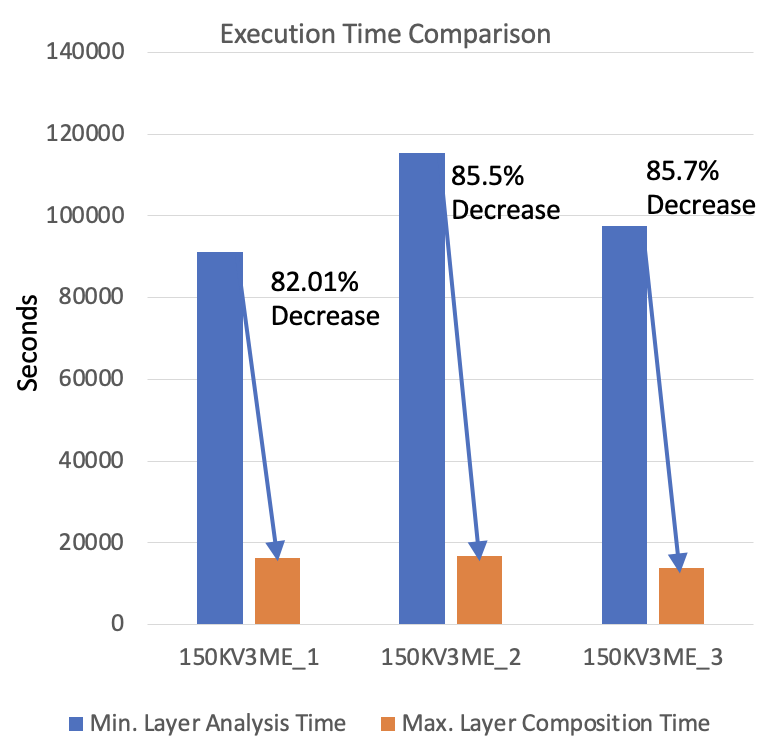}
        \caption{Min $\Psi$ vs. Max $\Theta$ Times}
        \label{cc1_performance_eval_sec5}
    \end{subfigure}
    \caption{Accuracy and Execution Time Comparison: \texttt{CC1} vs. Naive}
    \label{fig:cc1_accuracy_performance_eval}
\end{figure*}

\noindent Intuitively, CC nodes in single graphs have high degrees (\textit{more paths go through it and likely more shortest paths}) or have neighbors with high degrees (similar reasoning.) In the ANDed graph (ground truth graph), CC nodes that are common among both layers have a high chance of becoming a CC node. Moreover, if a common CC node has high overlap of neighborhood nodes from both layers with above average degree and low average sum of shortest path (SP) distances, there is a high chance of that node becoming a CC node in the ANDed graph. Using this intuition and observation, heuristic CC1 has been proposed for identifying CC nodes for two layers as a MLN. 

As discussed earlier, in the decoupling approach the analysis function $\Psi$ is used to analyze the layers once and the partial results and additional information in the composition function $\Theta$ are used to obtain intermediate/final results. For \texttt{CC1}, after the analysis phase ($\Psi$) on each layer (say, $G_x$) for each node (say, $u$), its degree ($deg_{x}(u)$) and sum of shortest path distances ($sumDist_{x}(u)$), and the one-hop neighbors ($NBD_{x}(u)$) are maintained if $u$ is a CC node, that is $u$ $\in$ $CH_x$. When calculating $sumDist_{x}(u)$, if a node $v$ is unreachable (which can happen if the graph/layer has multiple disconnected components), the distance $dist(u,v) = n$ where $n$ is the number of vertices in the layer (\textit{same in each layer -- HoMLN}) is considered. This is because the maximum possible path length in a graph with n nodes, is (n-1).


In the composition phase ($\Theta$), for each vertex $u$ in each layer (say, $G_x$), the degree distance ratio ($degDistRatio_{x}(u)$) is calculated with respect to the second layer (say, $G_y$) to estimate the likelihood of a node to be a CC node in the ANDed graph (ground truth graph) as per equation \ref{eq:ddr}.

\begin{equation}
\label{eq:ddr}
degDistRatio_{x}(u) =  \frac{sumDist_{x}(u)}{min(deg_{x}(u), deg_{y}(u))}
\end{equation} 

A smaller value of this ratio (i.e. smaller sum of distances and/or higher degree) means the vertex $u$ has a higher chance of becoming a CC node in the ANDed graph. 
When calculating the $degDistRatio_{x}(u)$, the degree is estimated for vertex $u$ in the ANDed graph, by taking into account the upper bound of degree, which is $min(deg_{x}(u), deg_{y}(u))$. Instead of using the degree of the nodes in each layer, using the estimated degree of ground truth graph gives a better approximation of the ratio value of the sum of distance and degree for the ground truth graph.
Due to \textit{decrease} in the edges in the ground truth graph, the average sum of distances will increase (as on an average, paths will get longer). As a result, the average degree distance ratio for the ANDed graph, $G_{xANDy}$ is estimated using equation \ref{eq:avg_ddr}.

\begin{equation}
\label{eq:avg_ddr}
\begin{split}
a&vgDegDistRatio_{xANDy}~~~~  = \\
& max(avgDegDistRatio_{x}, ~avgDegDistRatio_{y})
\end{split}
\end{equation}

For each CC node in each layer, the set of \textit{central} one-hop nodes, i.e. nodes having the $degDistRatio$ less than the $avgDegDistRatio_{xANDy}$, is calculated. Finally, those \textit{common} CC nodes from the two layers, which have a \textit{significant overlap} among the \textit{central} one-hop neighbors are identified as the CC nodes of the ANDed graph, that is given by the set $CH'_{xANDy}$. 
The complexity of the composition algorithm 
is dependent on the final step where the overlaps of CC nodes and their one-hop neighborhoods is considered. 
The algorithm will have a worst case complexity of $O(V^{2})$, if both layers are complete graphs consisting of V vertices. Based on the wide variety of data sets used in experiments, it is believed that the composition algorithm will have an average case complexity of $O(V)$. We are not showing the CC1 composition algorithm due to space constraints and decided to include a better composition algorithm (CC2.)

\begin{figure}[h]
  \centering
\includegraphics[width=1.1\linewidth]{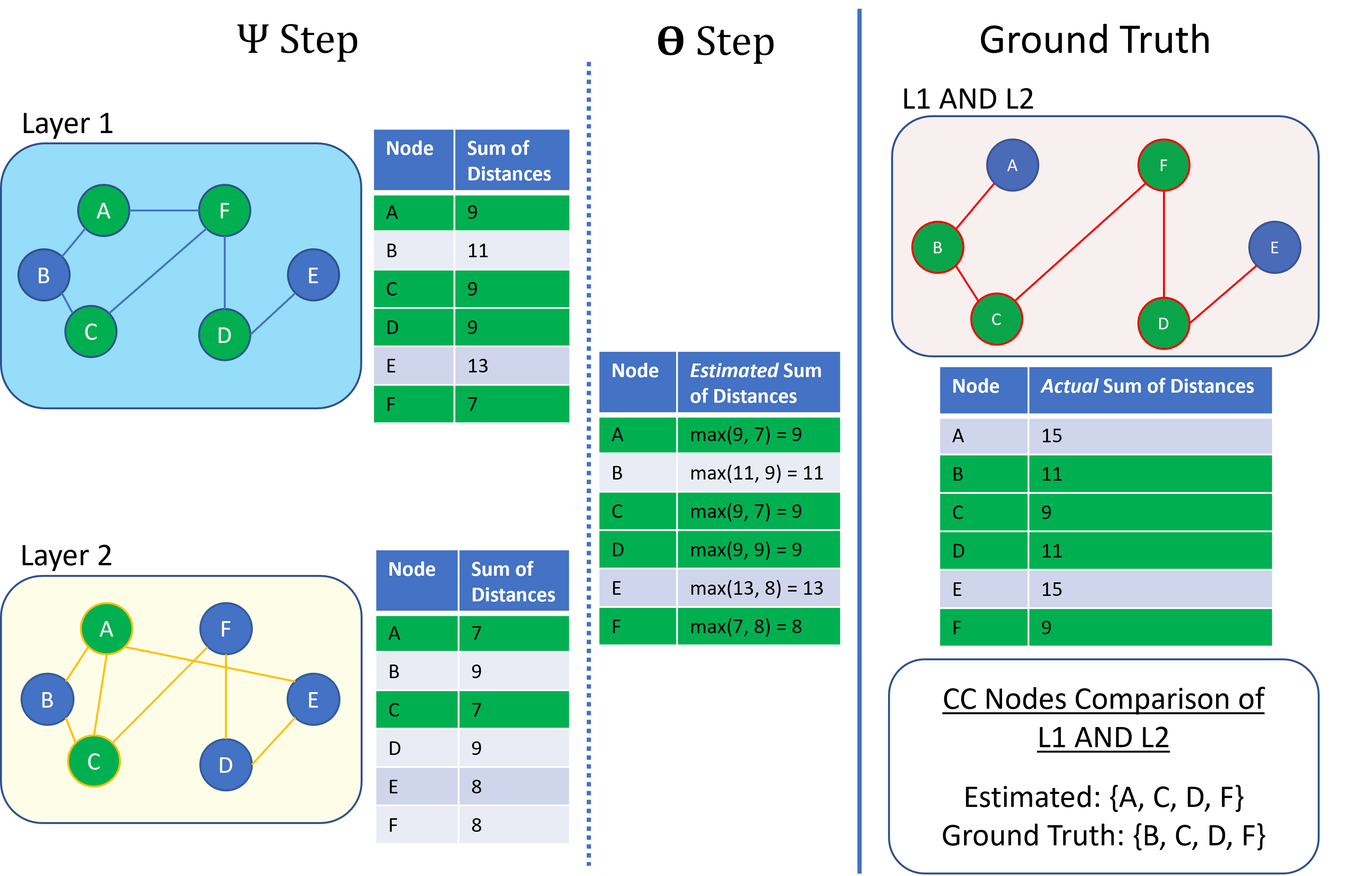}
  \caption{Intuition behind \texttt{CC2} with example.}
  \label{fig:cc2_intuition.png}
\end{figure}

\noindent \textbf{Discussion. } Experimental results for CC1 are shown in Section~\ref{sec:results}. 
Figure \ref{fig:cc1_accuracy_performance_eval} shows the accuracy (Jaccard coefficient) of \texttt{CC1} compared against naive approach on a subset of the synthetic data set 1 (details in Table \ref{tab:dataset1}). Although there is significant improvement in accuracy as compared to naive, it is still below 50\%. This is seen as a drawback of the heuristic. Also, composition time is somewhat high for large graphs. In addition, keeping one-hop neighbors of the CC nodes of both layers is significant amount of additional information. 
Figure \ref{cc1_performance_eval_sec5} shows the \textit{maximum composition time} against the \textit{minimum analysis time} of the layers (worst case scenario). Even though the composition time takes less time than the time it takes for the analysis of layers (one-time cost), the goal is to further \textit{reduce the composition time and the additional information necessary for the composition} without making any major sacrifice to accuracy. 

To overcome the above, \textit{Closeness Centrality Heuristic \texttt{CC2} has been developed, which keeps less information than \texttt{CC1}, has faster composition time, and provides better (or same) accuracy.}




\subsection{Closeness Centrality Heuristic CC2}

\noindent For heuristic \texttt{CC1}, the estimated value of $avgDegDistRatio_{xANDy}$ is less than or equal to the actual average degree-distance ratio of the AND aggregated layer, so false negatives will be generated. The composition step of \texttt{CC1} is also computationally expensive and requires a lot of additional information. Furthermore, \texttt{CC1} cannot identify all CC nodes of the ground truth graph. Closeness centrality heuristic 2 or \texttt{CC2} has been designed to address these issues.

The design of \texttt{CC2} is based on estimating the sum of shortest path (SP) distances for vertices in the ground truth graph. If the sum of SP distances of a vertex in the individual layers is known, the upper and lower limit of the sum of SP distances for that vertex in the ground truth graph can be estimated. This idea can be intuitively verified. Let the sum of SP distances for vertex $u$ in layer $G_x$ ($G_y$) be $sumDist_{x}(u)$ ($sumDist_{y}(u)$.) Let the set of layer $G_x$ ($G_y$) edges be $E_x$ ($E_y$.) In the ground truth graph, the upper bound for the sum of SP distances of a vertex $u$ is going to be $\infty$, if the vertex is disjoint in any of the layers. If $E_x \cap E_y = E_x$, the sum of SP distances for any vertex $u$ in the ground truth graph is going to be $sumDist_{x}(u)$. If $E_x \subseteq E_y$, then the ground truth graph will have the same edges as layer $x$ and sum of SP distances for a vertex $u$ in the ground truth graph will be same as the sum of SP distances for that vertex in layer $G_x$. When $E_x \cap E_y = E_x$, for any vertex $u$, $sumDist_{x}(u) \geq sumDist_{y}(u)$ because layer $G_x$ will have less edges than layer $G_y$ and average length of SPs in layer $G_x$ will be higher than average length of SPs in layer $G_y$. Similarly, when $E_x \cap E_y = E_y$, the sum of SP distances for any vertex $u$ in the ANDed graph is going to be $sumDist_{y}(u)$ and $sumDist_{y}(u) \geq sumDist_{x}(u)$. From the above discussion, it can be said that the sum of SP distances of a vertex $u$ in the ANDed graph is between $max(sumDist_{x}(u), sumDist_{y}(u))$ and $\infty$.

\begin{algorithm}
\small
\caption{Procedure for Heuristic CC2}
\label{alg:heuristic_2}

\noindent \textbf{INPUT:} $deg_x(u)$, $deg_x(u)$, $sumDist_x(u)$, $sumDist_y(u)$ $\forall u$ ; \\ $CH_x$, $CH_y$

\textbf{OUTPUT:} $CH'_{xANDy}$: estimated CC nodes in ANDed graph \\

\begin{algorithmic}[1]
\FOR{$u$}
    \STATE $estSumDist_{xANDy}(u) \gets max(sumDist_{x}(u), sumDist_{y}(u))$
\ENDFOR
\STATE Calculate $CH{'}_{xANDy}$ using $estSumDist_{xANDy}$
\end{algorithmic}
\end{algorithm}

Algorithm \ref{alg:heuristic_2} shows the composition step for the heuristic \texttt{CC2}. In this composition function, it is assumed that the estimated sum of SP distances of vertex $u$ is $estSumDist_{xANDy}(u)$. $sumDist_{x}(u)$ and $sumDist_{y}(u)$ are the sum of SP of node $u$ in layer $G_x$ and layer $G_y$ respectively, and $CH{'}_{xANDy}$ is the estimated set of CC nodes of the ANDed layer. Figure \ref{fig:cc2_intuition.png} illustrates how the composition function using CC2 is applied on a HoMLN with two layers. For each node $u$, the sum of SP is estimated, which is the maximum sum of SP of node among the two layers. In the example, node \textit{A} has a sum of SP as 9 and 7 in layer 1 and layer 2, respectively. The estimated sum of SP of node A will be 9 in the ANDed layer. Similarly, the sum of SP of other vertices can be estimated. Once the estimated sum of SP of all the vertices in the ANDed graph is completed, Equation \ref{eq:1} can be used to calculate the CC values of the nodes in the ANDed layer. As the CC value for the ANDed layer is calculated using the estimated sum of SP, one can either take the CC nodes with above-average closeness centrality scores or take the top-k CC nodes. For an MLN with $V$ nodes in each layer, the \textit{worst-case complexity of the composition algorithm for \texttt{CC2} is $O(V)$.}

\vspace{5pt}
\noindent \textbf{Discussion. } Experimental results for \texttt{CC2} are shown in Section~\ref{sec:results}. 
In general, this heuristic gave better recall as compared to the naive and \texttt{CC1} approach by being able to decrease the number of false negatives (shown in Table \ref{tab:recall_summary_syn}). Moreover, in terms of accuracy, both Jaccard coefficient and F1-score (shown in Table~\ref{tab:f1score_summary_syn}) are better for \texttt{CC2} as compared to Naive and CC1 approach, (except a few special cases details on which are in section \ref{sec:results}). Also it has been shown empirically that \texttt{CC2} is much more computationally efficient than \texttt{CC1}.



\vspace{-20pt}
\section{Experimental Analysis}

\label{sec:experimental_setup}
\noindent The implementation has been done in Python using the NetworkX \cite{hagberg2008exploring} package. The experiments were run on SDSC Expanse \cite{xsede} with single-node configuration, where each node has an AMD EPYC 7742 CPU with 128 cores and 256GB of memory running the CentOS Linux operating system.

\begin{table}[!h]
\caption{Summary of Synthetic Data Set-1}
\label{tab:dataset1}
\scriptsize
\centering
\renewcommand{\arraystretch}{1.4}
\begin{tabular}{|p{0.7cm}|p{0.25cm}|p{1cm}|c|c|c|}
\hline
{Base Graph} & \multirow{3}{1cm}{G$_{ID}$} & {Edge Dist. \%} &  \multicolumn{3}{c|}{\#Edges} \\ 
\cline{4-6}
 \#V, \#E & & \textit{L1\%,L2\%} & L1 & L2 & L1 AND L2 \\ 
\hline
\multirow{3}{0.7cm}{50KV, 250KE} & 1 & 70,30             & 224976   & 124988   & 50319           \\ 
\cline{2-6}
& 2 & 60,40 & 149982 & 199983   & 50392           \\ 
\cline{2-6}
& 3 & 50,50 & 174980 & 174977   & 50422           \\ 
\hline
\hline
\multirow{3}{0.7cm}{50KV, 500KE} & 4 & 70,30 & 399962   & 199986   & 51374           \\ 
\cline{2-6}
& 5     & 60,40 & 249978   & 349954   & 51458           \\ 
\cline{2-6}
& 6   & 50,50  & 299971   & 299960   & 51541           \\ 
\hline
\hline
\multirow{3}{0.7cm}{50KV, 1ME} & 7 & 70,30  & 349955   & 749892   & 55158           \\ 
\cline{2-6}
& 8 & 60,40   & 649918   & 449935   & 55647           \\ 
\cline{2-6}
& 9 & 50,50   & 549933   & 549922   & 55896           \\ 
\hline
\hline
\multirow{3}{0.7cm}{100KV, 500KE} & 10    &  70,30 & 249989   & 449970   & 100412          \\ 
\cline{2-6}
& 11     & 60,40               & 299986   & 399978   & 100494          \\ 
\cline{2-6}
& 12    & 50,50           & 349983   & 349981   & 100493          \\ 
\hline
\hline
\multirow{3}{0.7cm}{100KV, 1ME} & 13       & 70,30  & 799937   & 399978   & 101695          \\ 
\cline{2-6}
& 14 & 60,40       & 699948   & 499969   & 101822          \\ 
\cline{2-6}
& 15 & 50,50   & 599958   & 599964   & 101998          \\ 
\hline
\hline
\multirow{3}{0.7cm}{100KV, 2ME} & 16  & 70,30    & 699949   & 1499899  & 106389          \\ 
\cline{2-6}
& 17 & 60,40   & 1299914  & 899926   & 107141          \\ 
\cline{2-6}
& 18 & 50,50   & 1099924  & 1099923  & 107785          \\ 
\hline
\hline
\multirow{3}{0.7cm}{150KV, 750KE} & 19    & 70,30   & 674971   & 374979   & 150398          \\ 
\cline{2-6}
& 20 & 60,40   & 449982   & 599970   & 150447          \\ 
\cline{2-6}
& 21 & 50,50   & 524978   & 524975   & 150475          \\ 
\hline
\hline
\multirow{3}{0.7cm}{150KV, 1.5ME} & 22    & 70,30  & 1199942  & 599978   & 151684          \\ 
\cline{2-6}
& 23 & 60,40       & 749968   & 1049954  & 151883          \\ 
\cline{2-6}
& 24 & 50,50   & 899950   & 899956   & 152005          \\ 
\hline
\hline

\multirow{3}{0.7cm}{150KV, 3ME} & 25   & 70,30    & 1049951  & 2249888  & 156501          \\ 
\cline{2-6}
& 26 & 60,40   & 1349920  & 1949906  & 157295          \\ 
\cline{2-6}
& 27 & 50,50   & 1649922  & 1649909  & 157602          \\ 
\hline
\end{tabular}
\end{table}

For evaluating the proposed approaches, both synthetic and real-world-like data sets were used. Synthetic data sets were generated using PaRMAT \cite{parmat}, a parallel version of the popular graph generator RMAT \cite{rmat} which uses the Recursive-Matrix-based graph generation technique.

For diverse experimentation, for each base graph, 3 sets of synthetic data sets have been generated using PaRMAT. The generated synthetic data set consists of 27 HoMLNs with 2 layers with varying edge distribution for the layers. The base graphs start with 50K vertices with 250K edges and go up to 150K vertices and 3 million edges\footnote{Graph sizes larger than this could not be run on a single node due to the number of hours allowed and other limitations of the XSEDE environment.}. In the first synthetic data set, one layer (L1) follows \textit{power-law degree distribution} and the other one (L2) follows \textit{normal degree distribution}. In the second synthetic data set, both HoMLN layers have \textit{power-law degree distribution}. In the final synthetic data set, both layers have \textit{normal degree distribution}. For each of these, 3 edge distributions percentages (70, 30; 60, 40; and 50, 50) are used for a total of \textbf{81 HoMLNs of varying edge distributions, number of nodes and edges} for experimentation. 

Table \ref{tab:dataset1} shows the details of the different 2-layer MLNs from the first synthetic data set (L1: power-law, L2: normal) used in the experiments. The other two synthetic data sets have similar node and edge distributions.

For the real-world-like data set, the network layers are generated from real-world like monographs using a random number generator. The real-world-like graphs are generated using RMAT with parameters to mimic real world graph data sets as discussed in~\cite{chakrabarti2005tools}. 
As a result, the graphs have multiple connected components and also their ground truth graph.

\subsection{Result Analysis and Discussion}
\label{sec:results}

\noindent In this section, the results of the experiments have been presented. The two proposed heuristics have been tested on synthetic and real-world-like data sets with diverse characteristics. As a measure of accuracy, the Jaccard coefficient has been used. The precision, recall, and F1 scores for the proposed heuristics have also been compared. \textit{As a performance measure, the time taken by the decoupling approach has been compared with the time taken to compute the ground truth (as defined earlier in Section~\ref{sec:closeness_centrality}.) In addition,  the significance of the decoupling approach has also been highlighted by comparing the maximum composition time of the proposed algorithms with the minimum analysis time of the layers.} The accuracy of the algorithms is compared against the naive approach that serves as the baseline for comparison.



\begin{table}[h]
\caption{Accuracy Improvement of \texttt{CC1} and \texttt{CC2} over Naive}
\label{tab:accuracy_summary_syn}
\scriptsize
\centering
\renewcommand{\arraystretch}{1.4}
\begin{tabular}{|c|c|c|c|c|}
\hline
\multirow{2}{*}{Data Set} & \multicolumn{4}{c|}{Mean Accuracy}  \\
\cline{2-5}
 & \texttt{CC1} & \texttt{CC2} & \texttt{CC1} vs. Naive & \texttt{CC2} vs. Naive \\
\hline
Synthetic-1 & 43.56\% & 46.77\% & +52.57\% & {\bf \textcolor{blue}{+63.83\%}} \\
\hline
Synthetic-2 & 55.95\% & 55.20\% & {\bf \textcolor{blue}{+9.77\%}}  & +8.30\%  \\
\hline
Synthetic-3 & 48.87\% & 50.90\% & +47.55\% & {\bf \textcolor{blue}{+53.65\%}}
\\
\hline
Real-world-like & 88.71\% & 88.2\% & +7.36\% & +5.7\%
\\
\hline
\end{tabular}%

\end{table}

\noindent \textbf{Accuracy.} Figure \ref{fig:accuracy_syn_1} illustrates the accuracy of both the heuristics and the naive approach against the ground truth for the synthetic data set-1. As one can see, the accuracy of the heuristics is better than the naive approach in all cases. In most cases, \textbf{\texttt{CC2} performs better than \texttt{CC1}}. The accuracy of \texttt{CC1} increases with the graph density. A similar trend has been observed in other synthetic data sets as well, where the \textbf{proposed \texttt{CC1} and \texttt{CC2} heuristics perform better than the naive approach}.

\begin{figure*}
\centering
\includegraphics[width=0.8\textwidth]{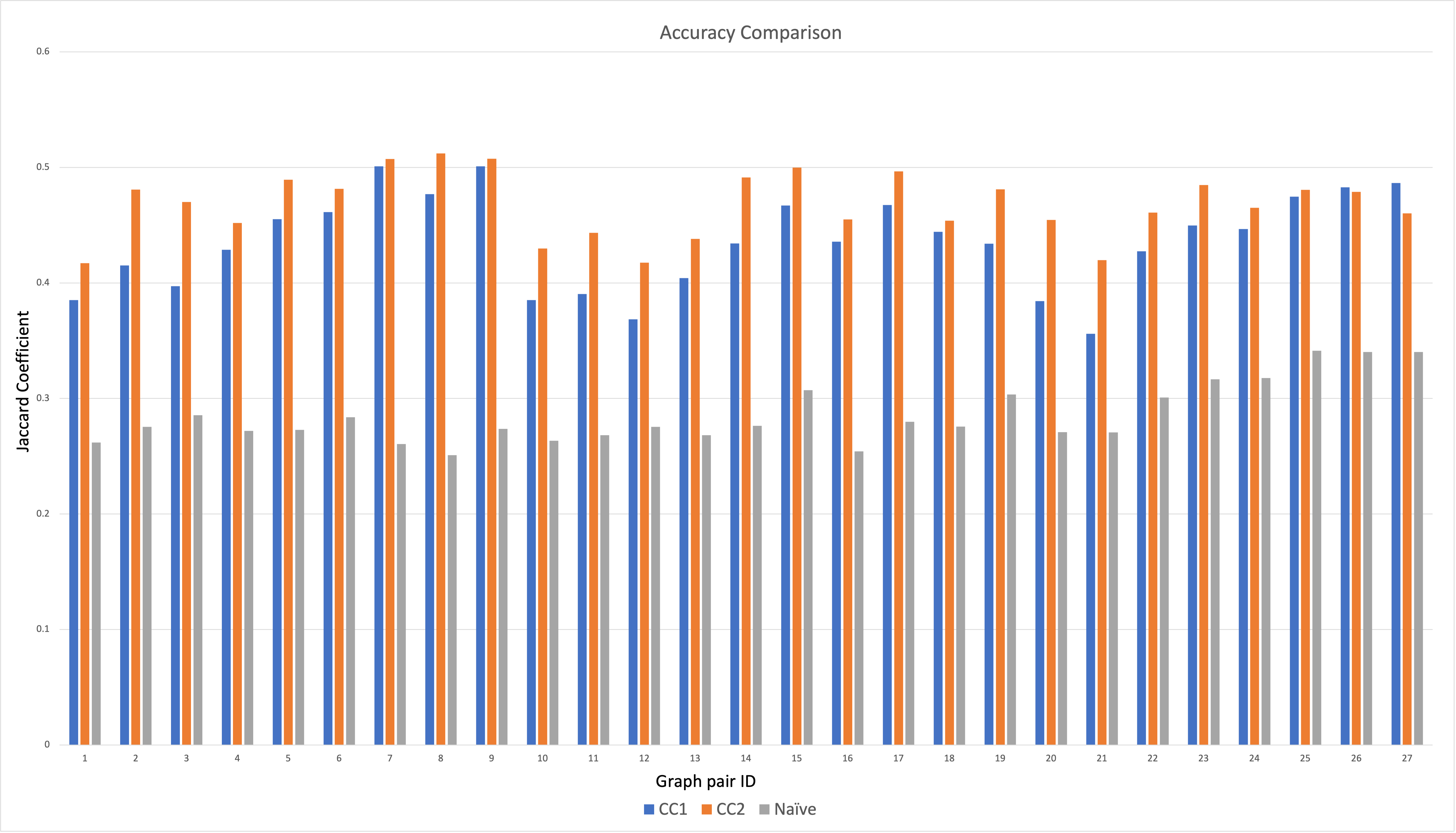}
\caption{Accuracy Comparison for Synthetic Data Set-1 (Refer Table \ref{tab:dataset1} where Graph Pair ID is $G_{ID}$)}
\label{fig:accuracy_syn_1}
\end{figure*}

Figure \ref{fig:accuracy_realworld} shows the accuracy of the algorithms on real-world-like data sets (distributions mimic real-world networks\cite{chakrabarti2005tools}). Across all data sets, both heuristics have \textbf{more than 80\% accuracy}. The accuracy of the heuristics does not go below the naive approach even for \textit{disconnected graphs}. This is significant, as the accuracy of the heuristics based on intuition is pretty good for real-world-like data sets. Although some of them show better accuracy for CC1 as compared to CC2, the efficiency improvement of CC2 is an order of magnitude better than CC1 (see Figure ~\ref{fig:performance_evaluation}).

\begin{figure*}
\centering
\includegraphics[width=0.8\textwidth]{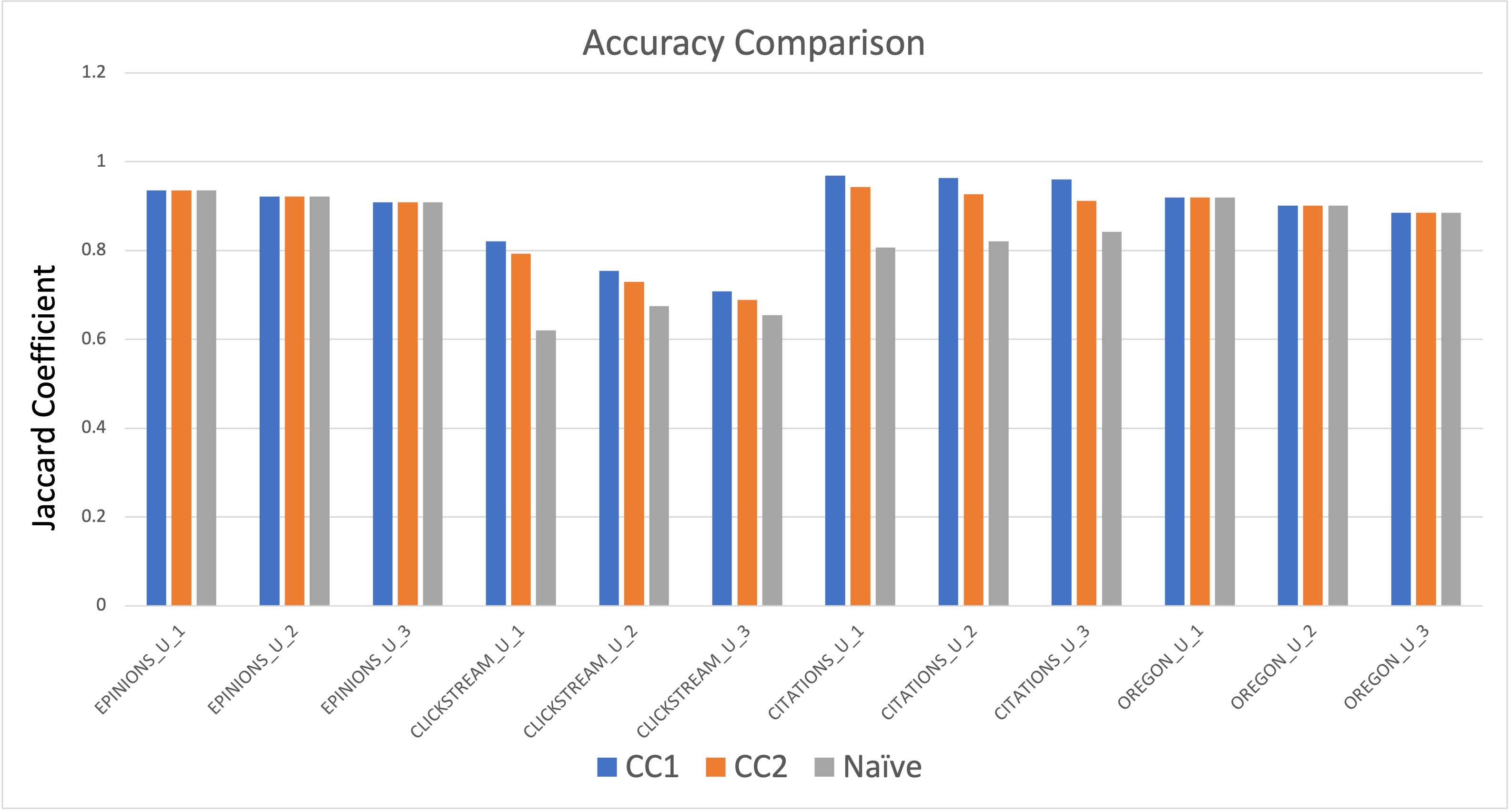}
\caption{Accuracy of the heuristics \texttt{CC1} and \texttt{CC2} for the real-world-like data sets} 
\label{fig:accuracy_realworld}
\end{figure*}

Table \ref{tab:accuracy_summary_syn} shows the mean accuracy and average percentage gain in accuracy over the naive approach for the synthetic data sets and the real-world-like data set. For all synthetic data sets, the proposed heuristics \textit{significantly} outperform the naive approach as shown. The least improvement in accuracy as compared to the naive approach is only when both layers have power-law degree distribution. Even for the real-world-like data set which has a high accuracy for the naive approach, the heuristics perform better than the naive approach.



\begin{table}[h]
\caption{Precision of \texttt{CC1} and \texttt{CC2} over Naive}
\label{tab:precision_summary_syn}
\scriptsize
\centering
\renewcommand{\arraystretch}{1.4}
\begin{tabular}{|c|c|c|c|c|}
\hline
\multirow{2}{*}{Data Set} & \multicolumn{4}{c|}{Mean Precision}  \\
\cline{2-5}
& \texttt{CC1} & \texttt{CC2} & \texttt{CC1} vs. Naive & \texttt{CC2} vs. Naive \\
\hline
Synthetic-1 & 60.98\% & 58.08\% & {\bf \textcolor{blue}{+1.52\%}} & -3.29\% \\
\hline
Synthetic-2 & 74.39\% & 72.02\% & -3.10\%  & -6.18\%  \\
\hline
Synthetic-3 & 51.99\% & 52.02\% & +0.006\% & {\bf \textcolor{blue}{+0.0718\%}}
\\
\hline
\end{tabular}%

\end{table}



\noindent \textbf{Precision.} After comparing the precision values received using the proposed heuristics against the ones from the naive approach for synthetic dataset-1, it is observed that
\texttt{CC1} has overall better precision compared to the naive approach and \texttt{CC2}. In general, it can be observed from Table \ref{tab:precision_summary_syn}, that across different types of HoMLNs, \texttt{CC1} and \texttt{CC2} give high precision values, ranging from \textbf{51\% to 74\%}. However, the improvement over naive is \textit{marginal} in most of the cases.



For synthetic data sets where one layer has normal degree distribution and the other layer has power-law degree distribution, \texttt{CC2} precision drops slightly. \texttt{CC2} was developed mainly to increase efficiency and preserve accuracy.



\begin{table}[h]
\caption{Recall Improvement of \texttt{CC1} and \texttt{CC2} over Naive}
\label{tab:recall_summary_syn}
\scriptsize
\centering
\renewcommand{\arraystretch}{1.4}
\begin{tabular}{|c|c|c|c|c|}
\hline
\multirow{2}{*}{Data Set} & \multicolumn{4}{c|}{Mean Recall}  \\
\cline{2-5}
& \texttt{CC1} & \texttt{CC2} & \texttt{CC1} vs. Naive & \texttt{CC2} vs. Naive \\
\hline
Synthetic-1 & 60.38\% & 71.01\% & +70.87\% & {\bf \textcolor{blue}{100\%}} \\
\hline
Synthetic-2 & 72.79\% & 71.38\% & {\bf \textcolor{blue}{+16.71\%}}  & +14.47\%  \\
\hline
Synthetic-3 & 48.87\% & 50.89\% & +47.55\% & {\bf \textcolor{blue}{+53.65\%}}
\\
\hline
\end{tabular}%

\end{table}

\vspace{5pt}



\noindent \textbf{Recall.} After comparing the recall values received using the proposed heuristics against the ones from the naive approach for synthetic data set 1, it is observed that \texttt{CC2} has overall better recall compared to the naive approach and \texttt{CC1}. In general, it can be observed from Table \ref{tab:recall_summary_syn} that across different types of HoMLNs, \textbf{both \texttt{CC1} and \texttt{CC2} have higher recall values than the naive approach}. Here, the heuristics achieve high recall values in the range from \textbf{49\% to 73\%}. It is observed that the proposed heuristics are able to decrease the false negatives more as compared to naive approach, which explains their marked improvement in terms of recall. However, the improvement over naive is \textit{marginal} in most of the cases. 



\begin{figure}[h]
\centering
\includegraphics[width=1.1\linewidth]{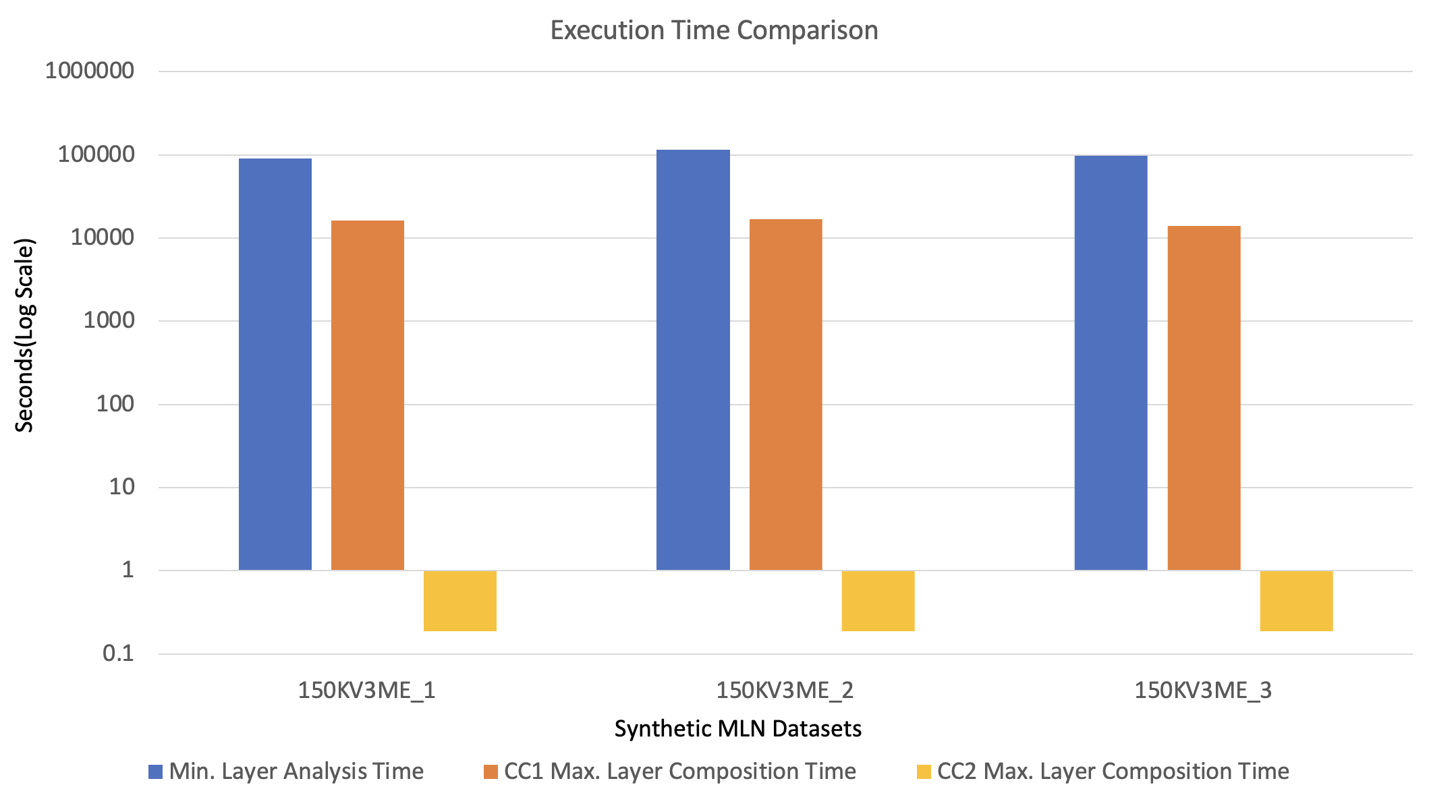}
\captionsetup{justification=centering}
\caption{Performance Comparison of CC1 and CC2 on Largest Synthetic Data Set \ref{tab:dataset1}: \\ Min. $\Psi$ Time vs. Max. \texttt{CC1} $\Theta$ Time vs. Max \texttt{CC2} $\Theta$ Time (worst case scenario)} \label{fig:performance_evaluation}
\end{figure}


\begin{table}[h]
\caption{F1-Score Improvement of \texttt{CC1} and \texttt{CC2} over Naive}
\label{tab:f1score_summary_syn}
\scriptsize
\centering
\renewcommand{\arraystretch}{1.4}
\begin{tabular}{|c|c|c|c|c|}
\hline
\multirow{2}{*}{Data Set} & \multicolumn{4}{c|}{Mean F1-Score}  \\
\cline{2-5}
& \texttt{CC1} & \texttt{CC2} & \texttt{CC1} vs. Naive & \texttt{CC2} vs. Naive \\
\hline
Synthetic-1 & 60.58\% & 63.80\% & +36.56\% & {\bf \textcolor{blue}{43.80\%}} \\
\hline
Synthetic-2 & 71.67\% & 71.09\% & {\bf \textcolor{blue}{+6.21\%}}  & +5.3\%  \\
\hline
Synthetic-3 & 50.22\% & 51.36\% & +24.71\% & {\bf \textcolor{blue}{+27.53\%}}
\\
\hline
\end{tabular}%
\end{table}



\noindent \textbf{F1-Score.} On comparing the F1-score received using the proposed heuristics against the ones from the naive approach for synthetic data set 1, it is observed that \texttt{CC2} has an overall better F1-score compared to the naive approach and \texttt{CC1}. It is established from Table \ref{tab:f1score_summary_syn} that across different types of HoMLNs, \textbf{both \texttt{CC1} and \texttt{CC2} have higher F1-scores than the naive approach}, reaching as high as \textbf{72\%}.

\vspace{5pt}

\noindent \textbf{Performance.} The ground truth graph obtained from Boolean AND operation on layers of HoMLN will always have same or less number of edges than the individual layers as an edge will appear in the ground truth graph only if it is connected between the same nodes in both the layers. The NetworkX \cite{hagberg2008exploring} package used here utilizes BFS to calculate the summation of distances from a node to every other node while calculating the normalized closeness centrality of the nodes. 
As the complexity of BFS depends on the number of vertices and edges in a graph, the ground truth will always require same or less time than the analysis time for the largest layer. 

Although the sum of the analysis time of the layers may be more than that of the ground truth, one needs to only consider the maximum analysis time of layers as they can be done in parallel.
Furthermore, the composition time is drastically less than the analysis time of any layer. Hence, \textit{the minimum analysis time for layers is compared with the maximum composition time to show the worst case scenario}. As can be seen from Figure \ref{fig:performance_evaluation} \textit{(plotted on log scale)}, the \textbf{maximum \texttt{CC1} composition time is at least 80\% faster and \texttt{CC2} is an order of magnitude faster than the minimum analysis time!}. In addition, the layer analysis is performed once and used for all subset CC node computation of n layers (which is exponential on n).



\noindent \textbf{Discussion.} Both proposed heuristics are better than the naive approach in terms of accuracy and way more efficient than ground truth computation. \texttt{CC2} is better than \texttt{CC1} if overall accuracy and efficiency are considered, but \texttt{CC1} performs better than \texttt{CC2} for high-density graphs and has better precision. The availability of multiple heuristics and their efficacy on accuracy and efficiency allows one to choose appropriate heuristics based on graph/layer characteristics.

\vspace{-8pt}

\vspace{-7pt}
\section{Conclusions and Future Work}
\label{sec:conclusions}
In this paper, the challenges of the decoupling approach for \textit{computing a global graph metric (closeness centrality) directly on a MLN }have been addressed. Two heuristics were developed to improve accuracy over the naive approach. \texttt{CC2} gives significantly higher accuracy than naive for graphs on a large number of  synthetic graphs generated with varying characteristics using RMAT and real-world-like data sets. CC2 is extremely efficient as well.





\vspace{-15pt}
\bibliographystyle{apalike}
{\small
\bibliography{./bibliography/pavelResearch,./bibliography/santraResearch,./bibliography/itlabPublications}}



\end{document}